\documentclass[12pt]{msml2020} 

\usepackage{algorithm}
\usepackage{algcompatible}
\usepackage{textcomp}
\usepackage{xcolor}
\def\BibTeX{{\rm B\kern-.05em{\sc i\kern-.025em b}\kern-.08em
    T\kern-.1667em\lower.7ex\hbox{E}\kern-.125emX}}
\usepackage{dsfont}
\usepackage{centernot}
\usepackage{bbm}

\newtheorem{observation}[theorem]{Observation}

\title[Scalable Influence Estimation Without Sampling]{Scalable Influence Estimation Without Sampling}
\usepackage{times}

\msmlauthor{%
 \Name{Andrey Y. Lokhov} \Email{lokhov@lanl.gov}\\
 \addr Theoretical Division, Los Alamos National Laboratory, Los Alamos, NM 87545
 \AND
 \Name{David Saad} \Email{d.saad@aston.ac.uk}\\
 \addr School of Engineering \& Applied Science, Aston University, Birmingham B4 7ET, UK
}

\begin{document}

\maketitle

\begin{abstract}%
  In a diffusion process on a network, how many nodes are expected to be influenced by a set of initial spreaders? This natural problem, often referred to as \emph{influence estimation}, boils down to computing the marginal probability that a given node is active at a given time when the process starts from specified initial condition. Among many other applications, this task is crucial for a well-studied problem of \emph{influence maximization}: finding optimal spreaders in a social network that maximize the influence spread by a certain time horizon. Indeed, influence estimation needs to be called multiple times for comparing candidate seed sets. Unfortunately, in many models of interest an exact computation of marginals is \#P-hard. In practice, influence is often estimated using Monte-Carlo sampling methods that require a large number of runs for obtaining a high-fidelity prediction, especially at large times. It is thus desirable to develop analytic techniques as an alternative to sampling methods. Here, we suggest an algorithm for estimating the influence function in popular \emph{independent cascade} model based on a scalable \emph{dynamic message-passing} approach. This method has a computational complexity of a single Monte-Carlo simulation and provides an upper bound on the expected spread on a general graph, yielding exact answer for treelike networks. We also provide dynamic message-passing equations for a stochastic version of the \emph{linear threshold} model. The resulting saving of a potentially large sampling factor in the running time compared to simulation-based techniques hence makes it possible to address large-scale problem instances.
\end{abstract}


\section{Introduction}
Accurate prediction of an outcome of information diffusion from a given set of initial spreaders is a challenging problem known as \emph{influence estimation}. This is one of the first natural questions that one would like to answer when studying a particular spreading process. Examples include estimation of the size of epidemic outbreak for determining the necessary quarantine and vaccination measures \cite{hethcote2000mathematics}; forecasting the impact of cascading failures in critical infrastructures for determining the set of control actions preventing the outage \cite{dobson2007complex};
or prediction of the outcome of a marketing or political campaign for an optimal use of limited budget and resources \cite{Domingos:2001:MNV:502512.502525}. This last application motivates a popular \emph{influence maximization} problem, pivotal for efficient marketing, opinion setting and other spreading processes within social networks, where the task is broadly defined as identifying a given number of individual constituents who will maximize the spread of information, or influence, within a certain time window~\cite{Domingos:2001:MNV:502512.502525}. This problem was first mathematically formulated for the Independent Cascade (IC) \cite{goldenberg2001talk} and Linear Threshold (LT) \cite{granovetter1978threshold} models by Kempe et al. in \cite{kempe2003maximizing}.
A remarkable result for the NP-hard seed selection problem states that a simple greedy algorithm guarantees a $(1-1/e)$ approximation to the optimal solution provided the oracle value for the influence function. The challenge is that the exact computation of the influence function has been proved to be \#P-hard~\cite{chen:2010:SIM:1835804.1835934} itself. Henceforth, in practice one typically resorts to approximation methods for carrying out the influence estimation task. Existing approaches to influence estimation can be broadly classified into two categories: sampling-based techniques and analytical methods, with or without rigorous guarantees. Without aiming at providing an extensive survey of existing methods, in the Related Work section below we point out most relevant references to the present work.

In this paper, we build on \emph{dynamic message-passing} (DMP) approach and develop two low complexity algorithms for an accurate estimation of the influence function in IC model in an arbitrary time frame: \textsc{DMPest} and \textsc{DMPinf}, for estimating the influence function at finite and infinite time horizons, respectively. \textsc{DMPest} and \textsc{DMPinf} correctly deal with the exclusion of the influence of the node to be updated, and estimate the spread without any restrictions on the values of transmission probabilities. Importantly, both algorithms scale only linearly with the number of edges in the graph for arbitrary initial conditions, including the probabilistic seeding assignment. We prove that DMP-estimated influence is exact on tree graphs and on graphs with the size of the loops that is larger than the time horizon of interest, and still provides an upper bound on exact influence on general loopy graphs. As the main focus of this work, we provide numerical evidence for the accuracy and scalability of the DMP-based approach, and conclude with a few remarks on perspectives of using DMP method for reducing the number of simulations in sampling-based approaches, as well as on possible extensions.

\section{Related work}

\subsection{Sampling methods}
Monte-Carlo sampling represents a natural class of methods used for estimating the influence function~\cite{kempe2003maximizing, Chen:2009:EIM:1557019.1557047, du2013scalable, cohen2014sketch, lucier2015influence, nguyen2017outward}. The crucial bottleneck within this approach consists in the time complexity of Monte-Carlo simulations that grows with the number of graph edges $\vert E \vert$, campaign deadline $T$ and the number of runs $R$ that are required to achieve a desired accuracy of influence estimate. The work \cite{du2013scalable} estimates the number of samples required to achieve a certain accuracy in continuous-time independent cascade model, which however depends on the actual unknown influence function. The sketching approach of \cite{cohen2014sketch} construct oracles guarantee constant-factor approximation approximations to the influence functions through samples of influence graphs, however the number of required samples can be quadratic in the size of the graph. Recent works \cite{lucier2015influence} and \cite{nguyen2017outward} provide advanced algorithms for reducing sampling requirements, and explore the possibility of distributed computations that is essential for scaling up the problem size. However, even with these improved algorithms, the required effective sampling factor $R$ scales at least linearly with the size of the network instance. Hence, it is natural to search for methods that might potentially save on these factors necessary for sampling methods.

\subsection{Analytical methods}
Given the potentially high cost of sampling methods, it is desirable to develop approximate influence estimation methods that do not rely on sampling, but still provide guarantees on the quality of the approximation. \textsc{SteadyStateSpread} algorithm for approximating influence based on a direct application of the model probability rule has been suggested in \cite{aggarwal2011flow}. However, this algorithm does not correctly account for the activation of the node to be updated, and hence is not exact even on tree graphs. This fact has been noticed in subsequent works \cite{zhang2013probabilistic} where inclusion-exclusion terms in the update step have been considered, and in \cite{yang2012approximation} where an approximation based on the linearization of the IC model update rule has been exploited. The paper \cite{zhou2015upper} provides an upper bound on the influence function through an approximate relation of spread at subsequent times via matrix of transmission probabilities. This bound constitutes a basis for the \textsc{UBLF} algorithm for influence maximization \cite{zhou2013ublf}. However, this approach suffers from similar shortcomings as equations in \cite{aggarwal2011flow}, and overestimates the spread even on tree graphs, while asymptotic approximation is obtained only for small transmission probabilities and large number of nodes in the graph. Moreover, convergence analysis of the series requires transmission probabilities to be strictly smaller than one, and the resulting algorithm requires inversion of a system-large matrix which is potentially time-consuming. Lastly, several other spectral methods have been developed for upper-bounding the influence at infinite time based on finding the spectral radius of the adjacency matrix \cite{draief2006thresholds}, Hazard matrix  \cite{lemonnier2014tight}, or its refinement accounting for the sensitive edges \cite{lee2016spectral}.

Several analytical approaches have been developed for particular cases or approximations to the IC model. For instance, \cite{kimura2006tractable} treats the variant of the IC model where propagation only occurs through shortest paths on the network. \cite{liu2014influence} provides an upper bound on the spread, but for the linear model where influence flowing into a node is a linear combination of influences flowing from its neighbors, instead of a product leading to non-linear expressions in the original IC model. Both approaches work as approximations for IC model only for small enough propagation probabilities. Exact equations for estimating influence of a single seed on a tree graph have been derived in \cite{jung2012irie}; however, the general case of an arbitrary set of active nodes in the approach of \cite{jung2012irie} requires the use of an oracle influence estimation algorithm or sampling as a subroutine.

The main challenge in the analytical approach is due to the fact influence estimation represents an NP-hard inference problem of computation of marginal probabilities in loopy graphical models. However, there exists a class of graphical model algorithms that are specifically tailored to tackle this problem. These techniques are related to loopy belief propagation method, and are commonly referred to as \emph{message-passing algorithms} \cite{Mezard:2009:IPC:1592967,Wainwright:2008:GME:1498840.1498841}. Recently, some of these methods have been generalized for selected dynamical models, leading to heuristic algorithms for approximating activation probabilities in several discrete and continuous-time epidemic, threshold and rumor spreading models \cite{PhysRevE.82.016101,altarelli2013large,PhysRevX.4.021024,shrestha2014message,PhysRevE.91.012811}. Message passing-type equations that are most intimately related to the infinite-time case in the present work have appeared in recent work focusing on the computation of the cascade size distribution for homogeneous \cite{burkholz2019efficient} and heterogeneous \cite{burkholz2019cascade} activation probabilities, and in \cite{abbe2017nonbacktracking} where similar expressions have been used for establishing upper bounds on the spread at infinite time using the non-backtracking walks approach, previously used for community detection problems \cite{krzakala2013spectral}. Finally, as noted in \cite{lokhov2016reconstructing}, the finite-time equations for the Independent Cascade model presented in this work are related to the dynamic equations for the susceptible-infected-recovered model \cite{volz2008sir,PhysRevE.82.016101,miller2011note,PhysRevX.4.021024,PhysRevE.91.012811} used for modeling some epidemic spreading processes, in the limit of deterministic recovery rate; as we will see below, a particular dynamical rule for the Independent Cascade model allows one to obtain simplified equations that contain only one type of dynamical variables, or messages, which make them practical for analysis and implementation.

\subsection{Additional Related Work}
A large body of work has been focused on improving the scalability of sampling methods, but acting as a subroutine for drastically reducing the number of unnecessary computations inside the greedy approach to the influence maximization problem; see~\cite{Leskovec:2007:COD:1281192.1281239, goyal2011celf++,Borgs:2014:MSI:2634074.2634144,Tang:2014:IMN:2588555.2593670,Tang:2015:IMN:2723372.2723734} for a non-exhaustive list. Most of these methods are related to influence estimation in the sense that they attempt to carefully prune the nodes that are not important from the influence maximization perspective, see \cite{doi:10.2200/S00527ED1V01Y201308DTM037} and \cite{v011a004} for recent reviews. However, they do not necessarily provide guarantees for the influence estimation problem considered here. Notice that in the present paper, we focus on the task of accurate influence estimation as a general problem that has applications for problems beyond influence maximization, for instance identifying the origin of the influence spread from measurements at sparsely located sensors \cite{shah2011rumors, lokhov2014inferring}, or estimation of model parameters from partially observed samples \cite{lokhov2016reconstructing, lokhov2015efficient}. In addition, there exists a large number of heuristic methods that aim at estimating the influence by completely neglecting the dynamic model, e.g. those based on variants of random selection, weighted degree distributions and node centralities. Although these techniques may scale well, we do not review them here as their performance is not guaranteed and significantly varies in quality depending on the considered model and setting. Finally, let us also mention recent papers that address the problem of direct estimation of influence functions from samples~\cite{NIPS2015_5989,du2014influence}. This line of research is in some sense orthogonal to the present contribution where we assume a well-defined model and develop an analytical method for scalable estimation of the influence function.

\section{Model and Problem Definition}
\label{sec:Problem}

\subsection{Model}
We study the popular model used in the studies of influence estimation and maximization, the discrete-time Independent Cascade model \cite{kempe2003maximizing}. An instance of the IC model is defined on a graph $G=(V,E)$ with $V$ and $E \subseteq V \times V$ denoting the set of nodes and pairwise edges, respectively. At any time, a node $i$ can be in either inactive or active state. A node $i$ activated at a time $t$ has a single chance to activate its neighbor $j$ at the subsequent time $t+1$ with probability $b_{ij}$ associated with the edge $(i,j) \in E$. Each realization of this diffusion process can be interpreted in terms of a live-edge graph \cite{kempe2003maximizing} defined on the set of nodes $V$ and described by a set of binary random variables $\mathbf{d} = \{ d_{ij} \}_{(i,j) \in E}$ associated with edges in the graph. Each edge $(i,j) \in E$ is declared \emph{live} randomly with probability $b_{ij}$ (in which case we set $d_{ij} = 1$), and \emph{blocked} otherwise (and thus characterized by $d_{ij} = 0$). Given the \emph{seeds}, i.e. a set $S$ of active nodes at initial time $t=0$, the set of eventually influenced nodes is given by the set of nodes that are connected to nodes in $S$ via the live edges. We will say that node $i$ is \emph{reachable} from node $j$ in time $t$ on graph $G$ if there exists a path connecting $i$ and $j$ such that $d_{ij} = 1$ along all edges of this path and its length is smaller or equal to $t$. Additionally, for the reasons that will become clear later, in the definition of reachable nodes we exclude paths that traverse the same edge multiple times.

Another model that has been considered in the context of influence estimation and maximization is the Linear Threshold model \cite{granovetter1978threshold}, where to each node $i$ is associated a certain activation threshold $\theta_i \in [0,1]$, and $i$ activates if the following condition is satisfied: $\sum_{j \in \partial i} b_{ji} x_{j} \geq \theta_i$, where $\partial i$ denotes the set of neighbors of $i$, entries of $\mathbf{b}$ satisfy $\sum_{j \in \partial i} b_{ji} \leq 1$, and $x_{j} = 1$ if $j$ is active and $x_{j} = 0$ otherwise. We do not focus on this model in the present work because the original version of this model assumes a deterministic dynamics \cite{granovetter1978threshold, kempe2003maximizing}, where influence can be easily estimated in linear time. It is however possible to generalize this model by introducing a stochastic activation rule and non-deterministic initial conditions. For completeness, in Appendix~\ref{app:LTM} we discuss the dynamic message-passing equations for the LT model, variants of which have previously appeared in the literature for different settings related to this model, see \cite{ohta2010universal,altarelli2013large,shrestha2014message,PhysRevE.91.012811}.

\subsection{Influence function}
The object of our main interest is the so-called \emph{influence function} $\sigma(S)$, that represents the expected number of ultimately influenced nodes averaged over the stochasticity of transmission probabilities $\mathbf{b} = \{ b_{ij} \}_{ij \in E}$, or, equivalently, over the realization of $\mathbf{d}$. As discussed above, this object is crucial in the classical \emph{influence maximization} problem that attempts to find the set of seeds of size $k$ leading to the maximum value of the influence. Although the influence function is essentially defined in the large time limit when the diffusion process stops, in many real-world scenarios with a very large number of nodes it might be desirable to predict the expected outcome at finite time horizon $T$, which may for instance represent the deadline of a marketing campaign \cite{du2013scalable}. Let $\sigma_{t}(S)$ be the expected number of active nodes at time $t$; with this definition, the original influence function $\sigma(S)$ can be equivalently denoted as $\sigma_{\infty}(S)$.

Classical formulations of the influence estimation and maximization problems assume a deterministic selection of the initial influence set of nodes. At the same time, it is easy to think of real-world applications with limited allocation budget, where an access to a subset of nodes for initial seeding is limited or influencing any desirable node is too costly; there is always a possibility that the initial seed might ``change its mind''. Instead, one may imagine a scenario of massive targeted advertisement campaign that attempts to reach specific nodes by spending more or less resources on implementing the initial seeding. Formally, assume that each node $i$ at initial time is activated independently with probability $p_i(0)$, so that the sum of these probabilities over the entire graph is equal to the total available budget: $\sum_{i \in V} p_i(0) = k$. In this setting, the influence function should be generalized to take into account an arbitrary initial condition $\mathbf{p}_{0} = (p_{1}(0), \ldots, p_{\vert V \vert}(0))$ of the type described above. It may therefore be advantageous to develop influence estimation method that can cope with this extension to probabilistic initial condition. Obviously, the case of classical seeds represents a particular case where probabilities $p_i(0)$ for respective nodes are set to one or zero.

\subsection{Marginal activation probabilities}
Let us denote $p_i(t)$ the probability that node $i$ is active at time $t$, and $t_i$ the time at which node $i$ gets activated (with a convention $t_i = T+1$ if node $i$ does not get activated before the time horizon $T$). Then, given the initial condition $\mathbf{p}_{0}$, the influence function at time $t$ can be expressed as a sum over marginal probabilities that nodes get activated before time $t$:
\begin{equation}
    \sigma_{t}(\mathbf{p_{0}}) = \mathbb{E} \left[\sum_{i \in V}\mathds{1}[t_{i} \leq t]\right] = \sum_{i \in V} p_i(t),
    \label{eq:exact_influence}
\end{equation}
where the expectation is taken over the realizations of $t_{i}$. Therefore, for producing an accurate estimation of the influence function $\sigma_{t}(\mathbf{p}_{0})$, it is crucial to estimate marginal probabilities $p_i(t)$. In the next section, we introduce our \emph{dynamic message-passing} approach to the computation of $p_i(t)$ in the IC model.

\section{Dynamic Message-Passing Method}
\label{sec:DMP}

In this section, we introduce dynamic message-passing equations that will serve as a foundation for our influence estimation algorithm. We will first start with a derivation on tree graphs, and will then state approximate equations that can be used in the case of loopy graphs for estimating marginal activation probabilities $p_i(t)$.

\subsection{Derivation of DMP equations on trees}

As a starting point, we notice that for each node $i \in V$ the marginal probability $p_i(t)$ can be exactly expressed as
\begin{equation}
    p_{i}(t) = 1 - \left[ 1 - p_{i}(0) \right] q_{i}(t).
    \label{eq:exact_marginal_p}
\end{equation}
The equation \eqref{eq:exact_marginal_p} is true for any graph, and straightforwardly conveys the following meaning: the probability that node $i$ is active at time $t$ is given by one minus the probability that $i$ has not been activated by time $t$. This last event is given by the probability $(1 - p_{i}(0))$ that node $i$ was not active initially times $q_{i}(t)$ that denotes the probability that node $i$ was not activated by its neighbors before time $t$.

Unfortunately, it is hard to compute the quantity $q_{i}(t)$ on a general graph. However, it is possible to compute $q_{i}(t)$ exactly on tree graphs, whereby
\begin{align}
q_{i}(t) = \prod_{j\in \partial i}q_{j \rightarrow i}(t),
\label{eq:tree_q_computation}
\end{align}
where $\partial i$ denotes the set of neighbors of $i$ on the graph $G$, and $q_{j \rightarrow i}(t)$ represents the probability that node $i$ did not get activated by its neighbor $j$ through the edge $(ji) \in E$ by time $t$. Note that by definition $q_{j \rightarrow i}(t)$ is conditioned on the fact that node $i$ is not active at time $t$.

Figure~\ref{fig:q_factorization}, Left provides an explanation why the expression \eqref{eq:tree_q_computation} is exact for tree graphs. Indeed, if node $i$ is not active, the influence of different branches centered in node $i$ are independent as on a tree by definition they do not have overlapping edges. Therefore, the expression for the marginal \eqref{eq:exact_marginal_p} can be rewritten as follows:
\begin{align}
p_{i}(t) = 1 - \left[ 1 - p_{i}(0) \right] \prod_{j\in \partial i}q_{j \rightarrow i}(t).
\label{eq:tree_marginal_p}
\end{align}

\begin{figure}[thb]
    \centering
    \includegraphics[width=0.45\columnwidth]{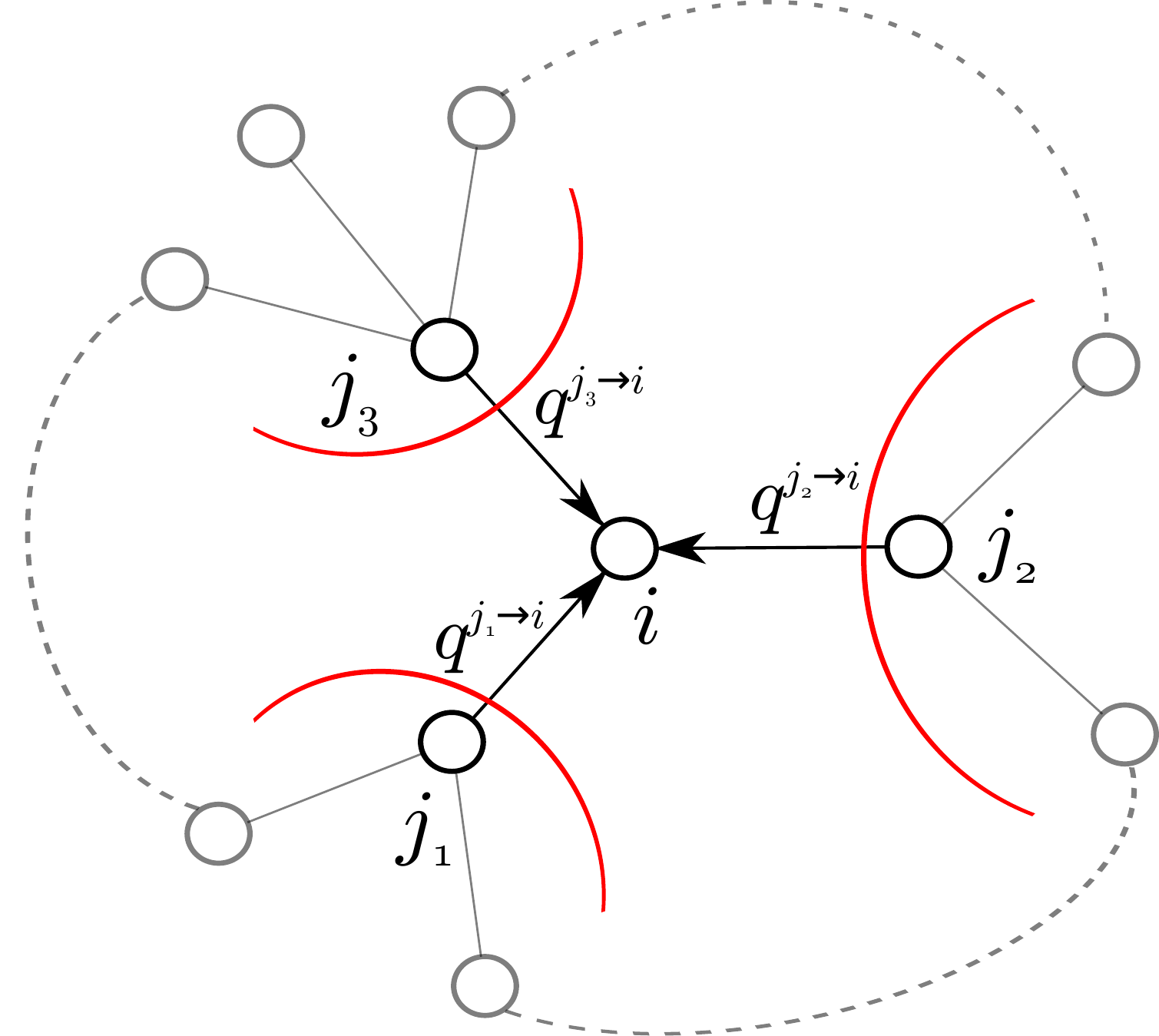} \hfill \includegraphics[width=0.39\columnwidth]{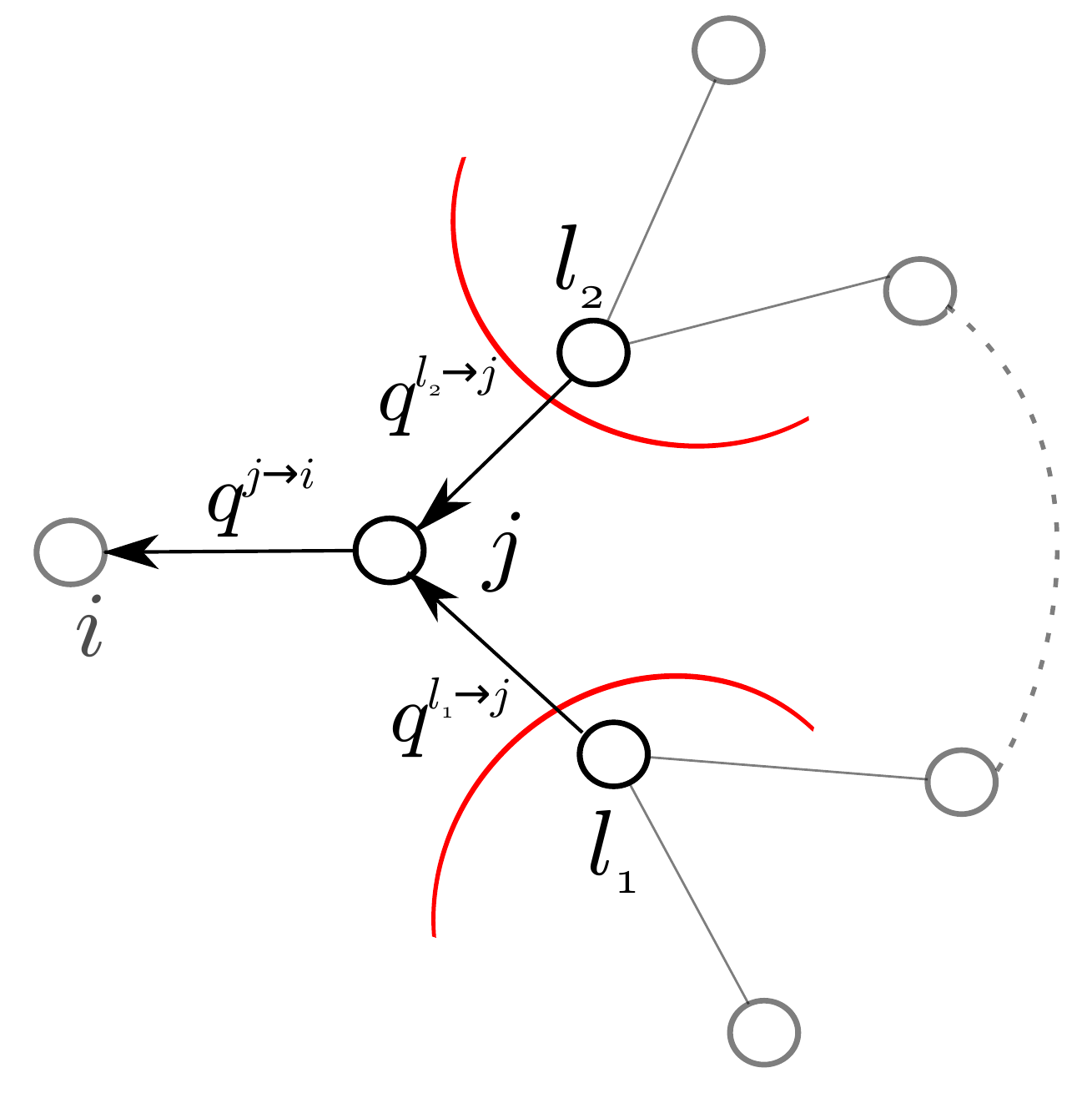}
    \caption{{\bf Left:} Illustration of computation of $q_i(t)$. On a tree (dashed paths not present), contribution of different branches are independent, and hence expressions \eqref{eq:tree_q_computation} and \eqref{eq:tree_marginal_p} are exact. On a loopy graph (dashed paths may exist), the quantities $q_{j \rightarrow i}(t)$ for $j \in \partial i$ are not independent, and factorization \eqref{eq:tree_q_computation} becomes only an approximation. {\bf Right:} Illustration of computation of $q^{(i)}_{j}(t)$. Node $i$ has been removed from the graph, and does not influence $j$. On a tree (dashed paths not present), contribution of different remaining branches are independent, and hence expressions \eqref{eq:tree_q_j_computation} and \eqref{eq:tree_p_through_q} are exact. On a loopy graph (dashed paths may exist), the quantities $q_{l \rightarrow j}(t)$ for $l \in \partial j \backslash i$ are not independent, and factorization \eqref{eq:tree_q_j_computation} becomes only an approximation. Both approximations \eqref{eq:tree_q_computation} and \eqref{eq:tree_q_j_computation} are good in treelike networks, where dashed paths are long, or in the case where transmission probabilities along the dashed path are small.}
    \label{fig:q_factorization}
\end{figure}

Let us now introduce a different quantity, $p_{j \rightarrow i}(t)$, that denotes the probability that $j$ is activated conditioned on the fact that $i$ is not. As previously, $p_{j \rightarrow i}(t)$ can be thought of as an equivalent of $p_{j}(t)$ defined on a graph where $i$ has been deleted together with all its adjacent edges. Therefore, similarly to \eqref{eq:exact_marginal_p}, $p_{j \rightarrow i}(t)$ can be defined as
\begin{equation}
    p_{j \rightarrow i}(t) = 1 - \left[ 1 - p_{j}(0) \right] q^{(i)}_{j}(t),
    \label{eq:exact_p}
\end{equation}
where $q^{(i)}_{j}(t)$ is defined in the cavity graph from which node $i$ has been removed. The definition \eqref{eq:exact_p} is valid for any graph, but, again, it is difficult to compute the quantity $q^{(i)}_{j}(t)$ on general graph. However, equivalently to \eqref{eq:tree_q_computation}, on tree graphs we have
\begin{align}
q^{(i)}_{j}(t) = \prod_{l\in \partial j \backslash i}q_{l \rightarrow j}(t).
\label{eq:tree_q_j_computation}
\end{align}
Equation \eqref{eq:tree_q_j_computation} mimics expression \eqref{eq:tree_q_computation}, but for $q^{(i)}_{j}(t)$ instead of $q_{i}(t)$. The difference is that the product in the right hand side of \eqref{eq:tree_q_j_computation} runs over $\partial j \backslash i$ that denotes the set of neighbors of $j$ except $i$; this is due to the definitions of $p_{j \rightarrow i}(t)$ and $q^{(i)}_{j}(t)$ that assume that $i$ is not active at time $t$ (or, alternatively, that $i$ has been removed from the graph), see Figure~\ref{fig:q_factorization}, Right for an illustration of this concept.


Now, similarly to \eqref{eq:tree_marginal_p}, the evolution of $p_{j \rightarrow i}(t)$ for times $t \geq 0$ on tree graphs simply reads:
\begin{equation}
p_{j \rightarrow i}(t)=1 - \left[ 1 - p_{j}(0) \right] \prod_{l\in \partial j \backslash i}q_{l \rightarrow j}(t).
\label{eq:tree_p_through_q}
\end{equation}

At this point, let us notice that $q_{j \rightarrow i}(t+1)$ can be expressed as $p_{j \rightarrow i}(t)$ by recalling the definition of these two quantities:
\begin{equation}
    q_{j \rightarrow i}(t+1) = 1 - b_{ji}p_{j \rightarrow i}(t).
    \label{eq:q_through_p}
\end{equation}
Indeed, the probability that $i$ did not get activated by its neighbor $j$ through the edge $(ji) \in E$ by time $t+1$ is given by one minus that both of the following events occurred: node $j$ got activated by time $t$ (this happens with probability $p_{j \rightarrow i}(t)$), and the edge $(ji)$ was live, i.e. $d_{ji} = 1$ (this happens with probability $b_{ji}$).

Now substituting relation \eqref{eq:q_through_p} in \eqref{eq:tree_p_through_q}, we finally obtain the following equations for $t>0$:
\begin{align}
&p_{j \rightarrow i}(t)=1 - \left[ 1 - p_{j}(0) \right] \prod_{l\in \partial j \backslash i}\left[ 1 - b_{lj} p_{l \rightarrow j}(t-1) \right], \label{eq:tree_p}
\\
&p_{i}(t) = 1 - \left[ 1 - p_{i}(0) \right] \prod_{j\in \partial i}\left[ 1 - b_{ji} p_{j \rightarrow i}(t-1) \right].
\end{align}
Once the conditional probabilities $p_{j \rightarrow i}(t)$ are obtained by running the system of equations on graph edges \eqref{eq:p}, the marginal probabilities $p_i(t)$ are estimated through equation \eqref{eq:marginal_p} .
Combined with the initialization
\begin{align}
p_{i \rightarrow j}(0) = p_{i}(0) \quad \forall (i,j) \in E,
\label{eq:tree_initialization}
\end{align}
equations \eqref{eq:tree_p}-\eqref{eq:tree_initialization} constitute our \emph{dynamic message-passing equations} for computing the marginals in IC model on tree graphs. The reason for this name is that \eqref{eq:tree_p} can be interpreted as passing ``dynamic messages'' (conditional probabilities $p_{j \rightarrow i}(t)$) along the edges of the graph.

\subsection{DMP equations on arbitrary graphs}

Definitions \eqref{eq:exact_marginal_p}, \eqref{eq:exact_p} and \eqref{eq:q_through_p} are valid on arbitrary graphs. However, equations \eqref{eq:tree_q_computation} and \eqref{eq:tree_q_j_computation} are no longer exact on graphs with loops: as explained in Figure~\ref{fig:q_factorization}, the factorization over branches rooted at node $i$ is in general no longer valid because quantities $q_{j \rightarrow i}(t)$ are not independent anymore due to potential presence of loops. However, we can still use these expressions as an approximation. A priori, this approximation is of a good quality as long as either of the two conditions are met: (i) the graph is locally treelike, i.e. the length of a typical loop is large; (ii) transmission probabilities are small. In both of these cases, the mutual influence between different $q_{j \rightarrow i}(t)$ becomes negligeble and vanishes while propagating through the loop. In the field of message-passing algorithms, these conditions are usually formally referred to as correlation decay properties, and can be proven for certain models \cite{Wainwright:2008:GME:1498840.1498841}. Condition (i) is typically met for random graphs for large enough $N$; for instance, the smallest loop in a sparse Erd\H{o}s-R\'enyi random graph with coordination number $c$ scales as $\log_c N$ \cite{Mezard:2009:IPC:1592967}.

Taking the considerations above into account, we essentially use equations \eqref{eq:tree_p} and \eqref{eq:tree_marginal_p} as a definition of the algorithm defined on a general graph with loops, replacing exact marginals $p_{i}(t)$ and messages $p_{j \rightarrow i}(t)$ by their estimates that we denote by $\widehat{p}_{i}(t)$ and $\widehat{p}_{j \rightarrow i}(t)$\footnote{This is similar to how loopy belief propagation algorithm relates to belief propagation algorithm derived on a tree graph \cite{Wainwright:2008:GME:1498840.1498841}.}. For instance, it means that we employ the following approximation for $p_{i}(t)$:
\begin{align}
p_{i}(t) \approx \widehat{p}_{i}(t) = 1 - \left[ 1 - p_{i}(0) \right] \prod_{j\in \partial i}\widehat{q}_{j \rightarrow i}(t),
\label{eq:approx_marginal_p}
\end{align}
where $\widehat{q}_{j \rightarrow i}(t) = 1 - b_{ji}\widehat{p}_{j \rightarrow i}(t)$ represents an estimate of $q_{j \rightarrow i}(t)$.

In the end of the day, DMP equations for IC model on arbitrary graphs read for all $i \in V$ and $(i,j) \in E$:
\begin{align}
&\widehat{p}_{j \rightarrow i}(t)=1 - \left[ 1 - p_{j}(0) \right] \prod_{l\in \partial j \backslash i}\left[ 1 - b_{lj} \widehat{p}_{l \rightarrow j}(t-1) \right], \label{eq:p}
\\
&\widehat{p}_{i}(t) = 1 - \left[ 1 - p_{i}(0) \right] \prod_{j\in \partial i}\left[ 1 - b_{ji} \widehat{p}_{j \rightarrow i}(t-1) \right].
\label{eq:marginal_p}
\end{align}
Estimates $\widehat{p}_{j \rightarrow i}(t)$ are obtained by running the system of equations on graph edges \eqref{eq:p} starting with the initialization
\begin{align}
\widehat{p}_{i \rightarrow j}(0) = p_{i}(0) \quad \forall (i,j) \in E.
\label{eq:initialization}
\end{align}
Marginal probabilities $p_i(t)$ are then estimated through equation \eqref{eq:marginal_p}.

At this point, let us make a connection to existing analytical results for the IC model in the literature. The \textsc{SteadyStateSpread} algorithm of \cite{aggarwal2011flow} is based on update equations that are similar to \eqref{eq:marginal_p}, except that $\widehat{p}_{j \rightarrow i}(t-1)$ in the right hand side of the equation is replaced by the estimate of the full marginal probability $\widehat{p}_{j}(t-1)$. Although the resulting expression is similar to the update rule in IC model for a single sampling run, it is easy to see that it leads to incorrect dependencies for activation probabilities even on a tree graph. In particular, it creates an ``echo chamber'' effect, where node $j$ can influence node $i$ while being influenced by $i$ at the previous step, a situation clearly prohibited by the model. The algorithms of \cite{kimura2006tractable} and \cite{zhou2015upper} can suffer from similar effects, while DMP equations correctly excludes this kind of effect via dynamic messages. Interestingly, this effect has been accounted for in an earlier work \cite{jung2012irie}, where equations for estimating influence of a single seed on a tree graph have been derived; however, the authors did not derive the respective equations for the case of an arbitrary set of active nodes, and instead wrote approximate equations that require the use of a different influence estimation algorithm or sampling as a subroutine. In the case of linearized scheme similar to the one in \cite{liu2014influence}, one should be careful and start from DMP equations \eqref{eq:p}-\eqref{eq:marginal_p} that provide a better approximation to the underlying marginal probabilities. It is worth noticing that most analytical methods reported better results in case of small transition probabilities. We anticipate that this may be possible due to the correlation decay of the type described above, and represents a natural behavior of approximation algorithms. Finally, as discussed in the Introduction, expressions similar to the large-time limit (discussed below) of the DMP equations introduced above have appeared as subroutines in \cite{abbe2017nonbacktracking,burkholz2019cascade}, while finite-time DMP equations are equivalent to those of the SIR model \cite{PhysRevE.82.016101,PhysRevE.91.012811} in the limit of deterministic recovery rate, but involve a singe dynamic message due to the peculiarity of dynamic rules of the IC model.

\section{Influence Estimation with DMP}
\label{sec:Estimation}

In this section, we use DMP equations derived in the previous section to estimate the influence function under IC model.

\subsection{Influence estimation at finite time}

DMP equations represent the central part of the \textsc{DMPest} algorithm that we use to provide an estimation $\widehat{\sigma}_{T}(\mathbf{p}_{0})$ of the influence function value \eqref{eq:exact_influence}, see Algorithm~\ref{alg:DMPest}. It is easy to see that for influence estimation at finite time horizon $T$, the computational complexity of \textsc{DMPest} is $O(\vert E \vert T)$, i.e. proportional to the complexity of a single Monte-Carlo simulation of the IC dynamics.

\begin{algorithm}[tb]
   \caption{\textsc{\textsc{DMPest} ($G$, $\mathbf{b}$, $\mathbf{p_{0}}$, $T$)}}
\begin{algorithmic}[1]
   \STATEx {\bfseries Input:} Graph $G = (V,E)$, time horizon $T$, transition probabilities $\mathbf{b} = \{ b_{ij} \}_{(ij) \in E}$, initial conditions $\mathbf{p}_{0} = (p_{1}(0), \ldots, p_{\vert V \vert}(0))$
   \FOR{$(i,j) \in E$}
   \STATE Initialize $\widehat{p}_{i \rightarrow j}(0) = p_{i}(0)$ as in \eqref{eq:initialization}
   \ENDFOR
   \FOR{$t=1$ {\bfseries to} $T-1$}
   \FOR{$(i,j) \in E$}
   \STATE Compute $\widehat{p}_{j \rightarrow i}(t)$ through \eqref{eq:p}
   \ENDFOR
   \ENDFOR
   \FOR{$i \in V$}
   \STATE Compute $\widehat{p}_{i}(T)$ using \eqref{eq:marginal_p}
   \ENDFOR
   \STATE $\widehat{\sigma}_{T}(\mathbf{p}_{0}) \leftarrow \sum_{i \in V} \widehat{p}_i(T)$
   \STATE {\bfseries return} $\widehat{\sigma}_{T}(\mathbf{p}_{0})$
\end{algorithmic}
\label{alg:DMPest}
\end{algorithm}

Let us now outline some properties of the DMP equations \eqref{eq:p}-\eqref{eq:initialization}. The following theorem suggests that influence estimate obtained through \textsc{DMPest} is exact when $G$ is a tree.

\begin{theorem}[Exact estimation on trees]
The influence estimate $\widehat{\sigma}_{T}(\mathbf{p}_{0})$ output by \textsc{DMPest} is exact if the underlying graph $G$ is a tree.
\label{th:tree}
\end{theorem}

\begin{proof}[Proof of Theorem \ref{th:tree}]
Directly follows from the derivation of DMP equations in tree networks in the previous section.
\end{proof}

Alternatively, the statement of the theorem follows from the fact that DMP equations \eqref{eq:p}-\eqref{eq:initialization} can be derived using the general belief propagation approach on time trajectories \cite{PhysRevE.91.012811} for models with progressive dynamics, while belief propagation equations provide exact marginal probabilities on tree graphs. Perhaps more interestingly, it is possible to show that \textsc{DMPest} algorithm provides an upper bound on the influence value on general loopy graphs, as stated in the following theorem.

\begin{theorem}[Bound on the influence on general graphs]
For general graphs, the estimate $\widehat{\sigma}_{T}(\mathbf{p}_{0})$ output by \textsc{DMPest} represents an upper bound on the influence function value $\sigma_{T}(\mathbf{p}_{0})$.
\label{th:bound}
\end{theorem}

The proof of the Theorem~\ref{th:bound} is given in the Appendix~\ref{app:proof_theorem2}. Theorem~\ref{th:bound} states that \textsc{DMPest} can be used to quantify the efficacy of investment in a marketing campaign, providing a guarantee that the spread will not rise above the bound given by \textsc{DMPest}. On the other hand, a careful examination of the proof of Theorem~\ref{th:bound} conveys that in some common cases (such as locally treelike networks or sparse random graphs) \textsc{DMPest} can provide exact predictions for the influence even in the presence of loops.

\begin{corollary}
    Let $L$ denote the length of the shortest loop in the graph $G$. \textsc{DMPest} provides an exact estimate of the influence if $L \geq 2T + 1$.
    \label{corollary}
\end{corollary}

In the spirit of the Corollary~\ref{corollary}, it is possible to use spanning trees to construct lower bounds on the influence function. Indeed, in practice, application of \textsc{DMPest} to any spanning tree of the original graph $G$ will provide such a lower bound as it will neglect possible correlations coming from the loops; a similar approach has been previously used in \cite{abbe2017nonbacktracking} using directed acyclic graphs.


\subsection{Influence estimation at infinite time}

DMP equations \eqref{eq:p}-\eqref{eq:initialization} have low algorithmic complexity compared to sampling methods, essentially saving a potentially very large multiplicative factor that is needed for gaining accuracy from simulations. Still, the algorithmic complexity of \textsc{DMPest} is linearly proportional to $T$. If one is interested in estimating the influence function to a certain precision $\epsilon$ at infinite time, there is no need to use \textsc{DMPest} with a large artificial bound on $T$: in fact, one can save this factor in computational complexity, as stated in the following Observation.

\begin{observation}[Influence estimation at infinite time]
    Influence at infinite time in the IC model can be estimated with an algorithmic complexity $O(\vert E \vert)$ only through the large time limit of DMP equations.
    \label{obs:DMPinf}
\end{observation}

Indeed, taking the  $t \to \infty$ limit in equations \eqref{eq:p}-\eqref{eq:initialization}, we immediately get the following fixed point equations for conditional probabilities:
\begin{align}
\widehat{p}_{j \rightarrow i}(\infty)=1 - \left[ 1 - p_{j}(0) \right] \prod_{l\in \partial j \backslash i}\left[ 1 - b_{lj} \widehat{p}_{l \rightarrow j}(\infty) \right], \label{eq:p_infty}
\end{align}
These equations can be solved self-consistently, for example by iteration until a certain tolerance $\epsilon$ on the difference between subsequent updates is met.
The marginal probabilities that are used to estimate the influence can be computed as follows:
\begin{align}
\widehat{p}_{i}(\infty) = 1 - \left[ 1 - p_{i}(0) \right] \prod_{j\in \partial i}\left[ 1 - b_{ji} \widehat{p}_{j \rightarrow i}(\infty) \right].
\label{eq:marginal_p_infty}
\end{align}
The pseudocode for the \textsc{DMPinf} algorithm based on equations that allows one to estimate influence at infinite time is given in Algorithm~\ref{alg:DMPinf}.

\begin{algorithm}[tb]
   \caption{\textsc{\textsc{DMPinf} ($G$, $\mathbf{b}$, $\mathbf{p_{0}}$, $\epsilon$)}}
\begin{algorithmic}[1]
   \STATEx {\bfseries Input:} Graph $G = (V,E)$, transition probabilities $\mathbf{b} = \{ b_{ij} \}_{(ij) \in E}$, initial conditions $\mathbf{p}_{0} = (p_{1}(0), \ldots, p_{\vert V \vert}(0))$, tol $\epsilon$
   \FOR{$(i,j) \in E$}
   \STATE Initialize $\widehat{p}_{i \rightarrow j}(\infty) = p_i(0)$
   \ENDFOR
   \WHILE{$\sum_{(j,i) \in E} \vert \widehat{p}^{\text{new}}_{j \rightarrow i}(\infty) - \widehat{p}^{\text{old}}_{j \rightarrow i}(\infty) \vert > \epsilon$}
   \FOR{$(i,j) \in E$}
   \STATE $\widehat{p}^{\text{new}}_{j \rightarrow i}(\infty) \leftarrow \{ \widehat{p}^{\text{old}}_{l \rightarrow j}(\infty) \}_{l\in \partial j \backslash i}$ through iteration \eqref{eq:p_infty}
   \ENDFOR
   \ENDWHILE
   \FOR{$i \in V$}
   \STATE Compute $\widehat{p}_{i}(\infty)$ using \eqref{eq:marginal_p_infty}
   \ENDFOR
   \STATE $\widehat{\sigma}_{\infty}(\mathbf{p}_{0}) \leftarrow \sum_{i \in V} \widehat{p}_i(\infty)$
   \STATE {\bfseries return} $\widehat{\sigma}_{\infty}(\mathbf{p}_{0})$
\end{algorithmic}
\label{alg:DMPinf}
\end{algorithm}

\section{Numerical Results}
\label{sec:Numerics}

In this section, our goal is to test the performance of \textsc{DMPest} and \textsc{DMPinf} in practice for real-world networks. Thanks to Theorem~\ref{th:tree}, we know that these algorithms are exact on tree graphs, so we do not present results on tree networks here. As explained in Section~\ref{sec:DMP}, existing heuristic approaches \cite{aggarwal2011flow, kimura2006tractable, zhou2015upper, liu2014influence, jung2012irie} do not provide exact answer even in the case of tree graphs. We will focus on testing \textsc{DMPest} and \textsc{DMPinf} on loopy graph instances; Monte Carlo simulations with a large sampling factor will be used as a benchmark tool to produce an estimate of the ground truth marginal probabilities. Notice that this is a more detailed information computed by our algorithms compared to a single sum representing the resulting influence \eqref{eq:exact_influence}.

We first illustrate the accuracy of \textsc{DMPinf} on a small real-world social network \cite{opsahl2009clustering}. We chose random transmission probabilities $\mathbf{b}$ distributed independently and uniformly at random for each edge in the interval $[0,1]$ in order to test the impact of heterogeneous parameters (that can be both arbitrarily small and large) on the accuracy of the message-passing procedure. The estimated marginal probabilities are plotted in Figure~\ref{fig:dmp-accuracy-Irvine}, Left against the ``ground truth'' obtained via $10^{6}$ Monte Carlo simulations for $k=20$ seeds chosen at random, and the topology of the network is sketched in Figure~\ref{fig:dmp-accuracy-Irvine}, Right.
It is remarkable that despite the presence of short loops, DMP shows an excellent agreement with the ground truth, saving a huge sampling factor for generating prediction compared to sampling approach.

In a synthetic experiment presented in Figure~\ref{fig:dmp-scaling}, we show that \textsc{DMPest} is indeed scalable to very large network instances with sizes beyond hundred of millions of nodes. Moreover, in agreement with the derivations in the Section~\ref{sec:Estimation}, we observe a linear-time scaling of \textsc{DMPest} with the number of edges in the network; in Figure~\ref{fig:dmp-scaling}, this is shown for the family of random regular graphs, where $\vert E \vert \propto \vert V \vert$.


\begin{figure}[htb]
    \centering
    \includegraphics[width=0.48\columnwidth]{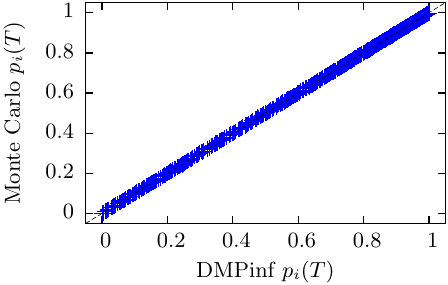}\includegraphics[width=0.37\columnwidth, angle=90]{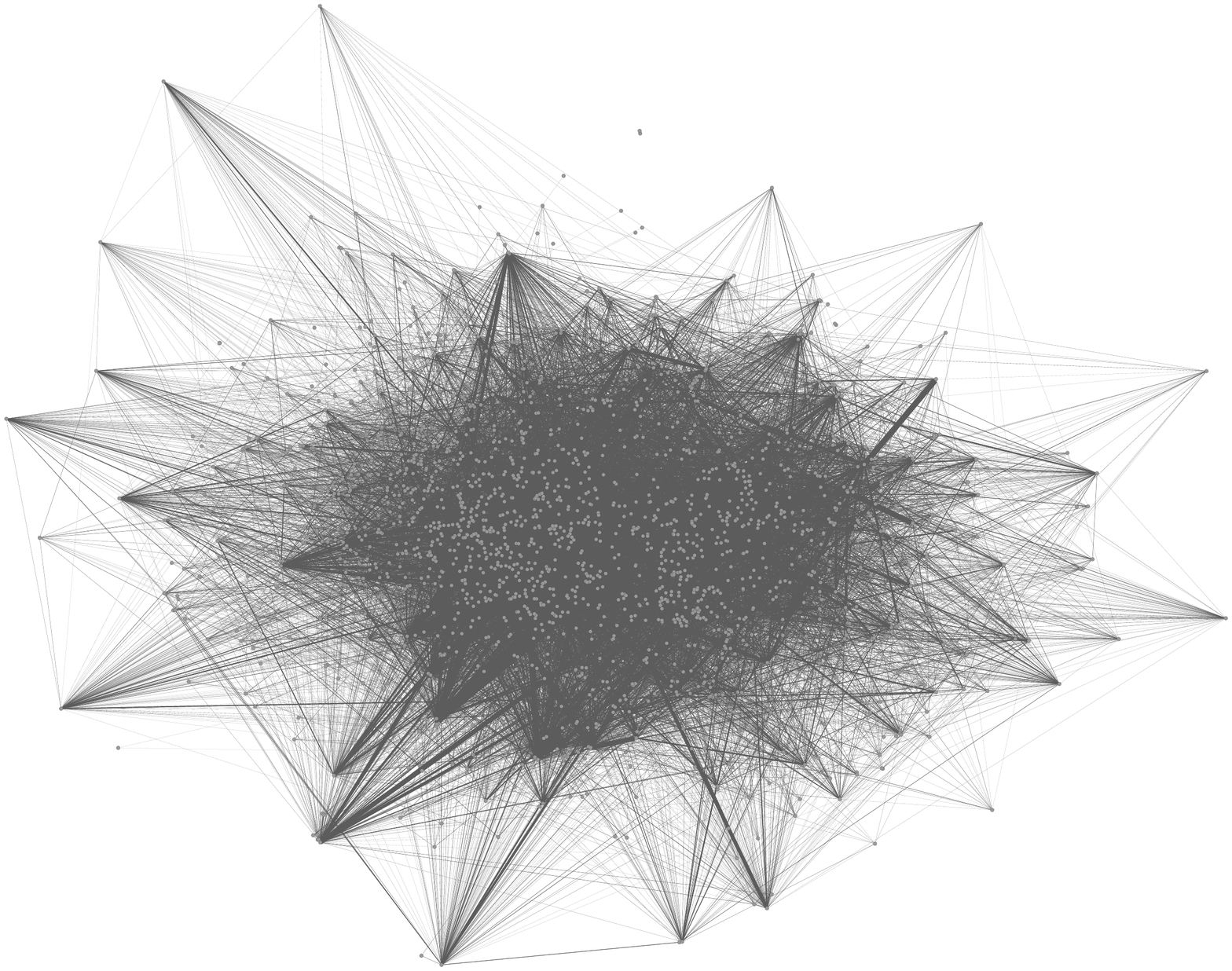}
    \caption{{\bf Left:} Scatter plot representing marginals predicted by \textsc{DMPinf} versus ``ground truth'' obtained by averaging $10^{6}$ Monte Carlo simulations for $T=\infty$ and $k=20$ randomly-selected seeds on a Facebook-like social network with $1899$ nodes and $20,296$ edges that represents an online community for students at University of California, Irvine \cite{opsahl2009clustering}. The per-node error $\Delta p_{i}(T) = \frac{1}{N} \sum_{j=1}^{N} \vert \widehat{p}_{j}(T) - p^{MC}_{j}(T) \vert$ is equal to $0.0044$ in this example. {\bf Right:} The topology of the social network used in simulations. In this representation, high-degree nodes are placed on the periphery. This network contains a large number of loops of short length.}
    \label{fig:dmp-accuracy-Irvine}
\end{figure}


\begin{figure}[htb]
    \centering
    \includegraphics[width=0.66\columnwidth]{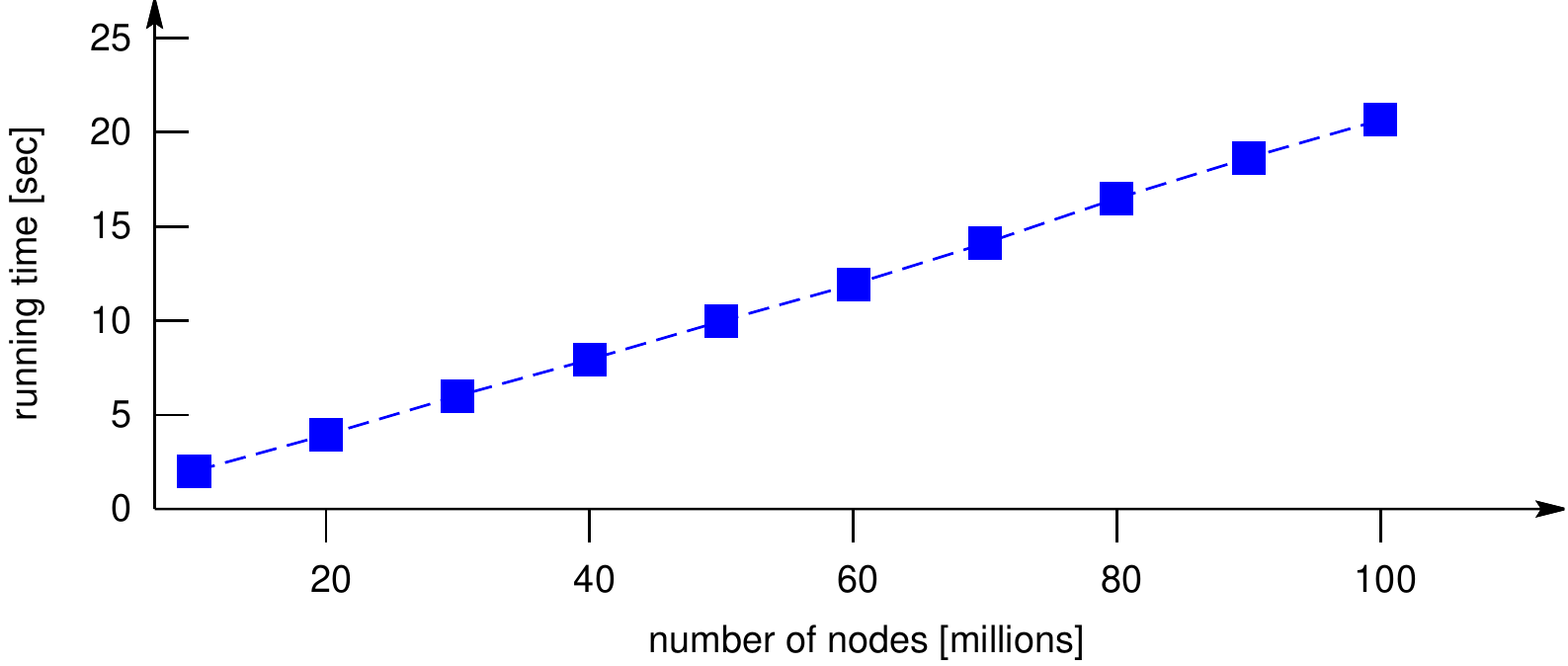}
    \caption{Linear computational time scaling of \textsc{DMPest} as a function of number of nodes obtained for $T=10$ on random regular graphs of degree $3$. Influence estimation on a graph of size $10^8$ takes only $20$ sec with \textsc{DMPest}, while MC simulations this task is essentially intractable.}
    \label{fig:dmp-scaling}
\end{figure}

\begin{table*}[hbt]
\caption{Test of accuracy of the \textsc{DMPest} algorithm on various real-world networks. The per-node error $\Delta p_{i}(T) = \frac{1}{N} \sum_{j=1}^{N} \vert \widehat{p}_{j}(T) - p^{MC}_{j}(T) \vert$ is given here in comparison with the marginal probability predictions obtained through $10^4$ Monte-Carlo runs. In these experiments, $b_{ij}$ are distributed randomly in $[0,0.1]$ on each edge $(ij) \in E$, $1\%$ ($0.01N$) of randomly placed seeds are selected at initial time, and $T=10$.}
\centering
\small
\setlength{\tabcolsep}{2pt}
\def\arraystretch{1.4}
\begin{tabular}{ |c|c|c|c|c|c| }
\hline
Network & $N$ & $M$ & $\Delta p_{i}(T)$ & DMP & $10^4$ MC  \\
name  &  nodes & edges  & error & runtime  & runtime  \\
\hline
UC Irvine social \cite{opsahl2009clustering} & 1899 & 20,296 & 0.003002 & 0.08 sec & 10.2 sec \\
\hline
GR collaborations \cite{leskovec2007graph}  & 5242 & 14,484 & 0.009101 & 0.04 sec & 18.8 sec \\
\hline
Internet autonomous \cite{nr} & 22,963 & 48,436 & 0.001593 & 1.15 sec & 1.6 min \\
\hline
Gnutella P2P \cite{Ripeanu:2002:MGN:613352.613670} & 62,586 & 147,892 & 0.000363 & 0.7 sec & 5.9 min  \\
\hline
Web-sk graph \cite{boldi2004ubicrawler} & 121,422 & 334,419 & 0.002523 & 1.7 sec & 20.1 min \\
\hline
Amazon co-purchasing \cite{yang2015defining} & 262,111 & 899,792 & 0.000469 & 6.1 sec & 68.3 min \\
\hline
\end{tabular}
\end{table*}

In Table~1, we provide an extensive benchmarking of the accuracy of \textsc{DMPest} on a number of real-world social and web networks.
For these tests, we had to limit the size of networks considered by hundred of thousands, in order to be able to run at least $10^4$ Monte-Carlo simulations for estimating the marginal probabilities. We see that \textsc{DMPest} yields impressively accurate results even on these graphs with loops. Notice that the sampling-based approach becomes prohibitively expensive already for graphs with tens of thousands of nodes, which makes it hard to use them in applications where influence estimation subroutine needs to be called many times.
At the same time, \textsc{DMPest} requires only a single run of the DMP equations, and thus results in extremely small running times for these real instances.

\section{Conclusions}
\label{sec:Conclusions}

In this paper, we presented an analytical approach to influence estimation based on dynamic message-passing approach. This method provides an estimation of the influence function with an algorithmic complexity $O(\vert E \vert T)$ for finite-time horizon problems and $O(\vert E \vert)$ for influence estimation at infinite time, which makes it possible to apply the developed algorithms to large-scale problems. Importantly, developed algorithms provide an exact estimation of the influence function on tree graphs, and an upper bound on the influence value on general loopy graphs. DMP-based algorithms should be especially accurate on sparse locally treelike graphs due to diverging size of loops. These practical aspects of the DMP algorithm were at the focus of this work: demonstration of an excellent performance of the algorithm on a variety of real-world network instances, both in terms of the quality of predictions and of the computational complexity.

Due to the upper bounds provided by \textsc{DMPest} and \textsc{DMPinf}, these algorithms can be straightforwardly used in applications where estimation needs to be called many times, e.g. for pruning nodes that have a week influence spreading potential. Corollary~\ref{corollary} suggests an even more interesting use of \textsc{DMPest} in conjunction with MC sampling methods. Indeed, we know that at finite times \textsc{DMPest} will provide exact influence estimation for nodes that have a treelike neighborhood, and will only make a mistakes in the regions with loops. Therefore, it is possible to use \textsc{DMPest} to save a potentially significant number of samples for providing very accurate influence estimation in sparse regions, and use MC sampling in dense regions of the graph. Moreover, lower bounds that can be obtained in practice by running \textsc{DMPest} and \textsc{DMPinf} on spanning trees of the graph, can further reduce the use of MC simulations if the DMP-estimated marginals appear to be tight. Empirical exploration of these directions is left for future work.

An explicit algorithmic form of the influence functions paves a way towards developing new heuristic DMP-based optimization algorithms for influence maximization problems, similar to how message-passing equations have been used for other models in previous work \cite{altarelli2013large,LokhovE8138}. In particular, it would be interesting to explore the settings of non-deterministic seeding as DMP equations are valid for arbitrary factorized probabilistic initial condition, which generalizes the case of fixed seeds. Future work should also focus on the development of robust DMP-based framework for dealing with uncertainty in parameters, which should be particularly useful for applications to the robust version of the influence maximization problem \cite{He:2016:RIM:2939672.2939760,Chen:2016:RIM:2939672.2939745}.

\acks{The research presented in this work was supported by the Laboratory Directed Research and Development program of Los Alamos National Laboratory under project number 20200121ER.}

\bibliography{biblio}

\appendix

\section{Dynamic message-passing equations for the stochastic linear threshold model}
\label{app:LTM}

As discussed in Section~\ref{sec:Problem}, another popular model considered in the context of influence estimation and maximization is the Linear Threshold model. The deterministic version of LT model described in does not present any difficulty from the influence estimation perspective, as running a single simulation is sufficient for evaluating the influence function. However, an algorithm that allows for an analytical estimation of the influence even in the deterministic case can still be relevant for inference, optimization, or learning applications; message-passing equations for the this version of the LT model appeared in \cite{altarelli2013large}, and have essentially the same structure as dynamical equations previously derived for an equivalent (in a certain limit) model from statistical physics, zero-temperature random field Ising model ($T=0$ RFIM), see \cite{ohta2010universal} for more details.  Here, we consider a stochastic version of this model, where the activation of node $i$ happens with probability $\eta_{i}$ when the threshold condition $\sum_{j \in \partial i} b_{ji} x_{j} \geq \theta_i$ is satisfied, and for potentially stochastic initial condition, similarly to the setting discussed above for the IC model. Dynamic message-passing equations for this generalized model have been studied in \cite{shrestha2014message} for continuous time and in \cite{PhysRevE.91.012811} for discrete time (in the form of the equivalent $T=0$ RFIM). Here, for completeness we state the DMP equations for the discrete-time LT model as defined above. The notations are equivalent to the ones used in the DMP equations for the IC model.
\begin{align}
    \nonumber
    p_{i}(t+1) & = (1 - \eta_{i}) p_{i}(t)
    \\
    & + \eta_{i} \sum_{\{x_k\}_{k \in \partial i}} \hspace{-0.2cm} \mathbbm{1} \left[ \sum_{k \in \partial i} b_{ki}x_k \geq \theta_i \right] \prod_{k \in \partial i: x_k=1} \hspace{-0.1cm} p_{k \to i}(t) \hspace{-0.1cm} \prod_{k \in \partial i: x_k=0} \hspace{-0.1cm} \left[ 1 - p_{k \to i}(t) \right];
    \label{eq:LTM_marginal}
    \\
    \nonumber
    p_{i \to j}(t+1) & = (1 - \eta_{i}) p_{i \to j}(t)
    \\
    & + \eta_{i} \sum_{\{x_k\}_{k \in \partial i \backslash j}} \hspace{-0.2cm} \mathbbm{1} \left[ \sum_{k \in \partial i \backslash j} b_{ki}x_k \geq \theta_i \right] \prod_{k \in \partial i \backslash j: x_k=1} \hspace{-0.1cm} p_{k \to i}(t) \hspace{-0.1cm} \prod_{k \in \partial i \backslash j: x_k=0} \hspace{-0.1cm} \left[ 1 - p_{k \to i}(t) \right],
    \label{eq:LTM_message}
\end{align}
supplemented with the initial condition $p_{i \to j}(0) = p_{i}(0)$ for all $i$ and $j$, where $p_{i}(0)$ is a (in general stochastic) initial condition for the node $i$.

Similarly to the case of the DMP equations for the IC model, equations \eqref{eq:LTM_marginal}-\eqref{eq:LTM_message} are exact on tree graphs, see \cite{PhysRevE.91.012811} for details. Interestingly, it is shown in \cite{shrestha2014message} that no upper bound through mechanism discussed in Theorem~\ref{th:bound} exists for the stochastic LT model, although empirical evaluation of the predictions of the DMP equations still shows good agreement when compared to the marginal probabilities of the model. The work \cite{khim2016computing} derives bounds for the LT model based on the spectral bound approaches, similar to the ones for the IC model discussed in the Introduction.

\section{Proof of Theorem 2}
\label{app:proof_theorem2}

Before proceeding with the proof of the Theorem~\ref{th:bound}, let us give an equivalent representation for the probability \eqref{eq:exact_marginal_p} in terms of a live-edge graph:
\begin{equation}
    q_{i}(t) = \left\langle \prod_{l \in \mathcal{N}^i_t[\mathbf{d}]} \left[1 - p_{l}(0) \right] \right\rangle
    \label{eq:q_counting}
\end{equation}
Here, $\mathcal{N}^i_t[\mathbf{d}]$ denotes the set of nodes from which node $i$ is reachable in time $t$ given a particular realization of $\mathbf{d}$, and $\left\langle \cdot \right\rangle$ is an average with respect to the realizations of $\mathbf{d}$. Equation \eqref{eq:q_counting} has the following meaning: the probability that node $i$ did not get activated by time $t$ is given by the average over realizations in which all nodes that are reachable from $i$ were not active at initial time.

Similarly, in the live-edge representation, the probability $q_{j \rightarrow i}(t)$ can be expressed as follows:
\begin{equation}
    q_{j \rightarrow i}(t) = \left\langle \prod_{l \in \mathcal{N}^{i \leftarrow j}_t[\mathbf{d}]} \left[1 - p_{l}(0) \right] \right\rangle,
    \label{eq:q_message_counting}
\end{equation}
where $\mathcal{N}^{i \leftarrow j}_t[\mathbf{d}]$ is the set of nodes from which node $i$ is reachable in time $t$ given a particular realization of $\mathbf{d}$ with an additional constraint that the corresponding path ends on the edge $(ji) \in E$. Recall that by our definition we exclude paths that cross $i$: first arriving from another neighbor $l \in \partial i$ and $l \neq j$, reaching $j$ and then coming back to $i$, which is consistent with $q_{j \rightarrow i}(t)$ being conditioned on $i$ being non-active. Alternatively, one can think of $\mathcal{N}^{i \leftarrow j}_t[\mathbf{d}]$ as a reachable set defined on a \emph{cavity} graph where all edges outgoing from $i$ have been deleted.

The proof of Theorem \ref{th:bound} also makes use of the following two Lemmas that we state below.

\begin{lemma}[Chebyshev integral inequality]
Let $f_1(x_1, \ldots, x_n)$, $\ldots$, $f_m(x_1, \ldots, x_n)$ be comonotonic functions, i.e. simultaneously non-increasing or non-decreasing in each of their arguments. Then
\begin{equation}
    \left\langle \prod_{i=1}^{m} f_i(x_1, \ldots, x_n) \right\rangle \geq \prod_{i=1}^{m} \left\langle f_i(x_1, \ldots, x_n) \right\rangle,
\end{equation}
where $\left\langle \cdot \right\rangle$ denotes an average over the distribution of random variables $x_1, \ldots, x_n$.
\label{lemma:chebyshev}
\end{lemma}

\begin{lemma}[Overestimation of $p_{j \rightarrow i}(t)$]
On general graphs with loops, estimates obtained through \eqref{eq:p} satisfy $\widehat{p}_{j \rightarrow i}(t) \geq p_{j \rightarrow i}(t)$ for all $(i,j) \in E$.
\label{lemma:underestimation}
\end{lemma}

The proofs of Lemmas \ref{lemma:chebyshev} and \ref{lemma:underestimation} are given in the Appendix~\ref{app:proofs_lemmas}. The proof technique is similar to \cite{PhysRevE.82.016101}.

\begin{proof}[Proof of Theorem \ref{th:bound}]



The proof starts with the exact expression \eqref{eq:exact_marginal_p}, and unfolds with the analysis of each of the approximation steps in \eqref{eq:approx_marginal_p}, \eqref{eq:p} and \eqref{eq:marginal_p}, making use of the live-edge graph representation. Let us first notice that using definitions of the sets $\mathcal{N}^i_t[\mathbf{d}]$ and $\mathcal{N}^{i \leftarrow j}_t[\mathbf{d}]$, we have for all realizations of random variables $\mathbf{d}$:
\begin{equation}
     \vert \mathcal{N}^i_t[\mathbf{d}] \vert \leq \sum_{j \in \partial i} \vert \mathcal{N}^{i \leftarrow j}_t[\mathbf{d}] \vert.
     \label{eq:overcounting_sets}
\end{equation}
An illustration for this observation is given in Figure~\ref{fig:reachability}. Assume that $t=3$, and that the realization of $\mathbf{d}$ is the one shown in Figure~\ref{fig:reachability}: all edges in the vicinity of $i$ are live except the edge $(i,j_3)$. By definition, the reachable sets read for this case: $\mathcal{N}^i_{t=3}[\mathbf{d}] = \{j_1, j_2, l_1, l_2\}$; $\mathcal{N}^{i \leftarrow j_1}_{t=3}[\mathbf{d}] = \{j_1, l_1, l_2\}$;  $\mathcal{N}^{i \leftarrow j_2}_{t=3}[\mathbf{d}] = \{j_2, l_1, l_2\}$; and $\mathcal{N}^{i \leftarrow j_3}_{t=3}[\mathbf{d}] = \emptyset$. Notice that $l_1$ and $l_2$ appear in both sets reachable through edges $(i,j_1)$ and $(i,j_2)$, and therefore . On this example, we see that as long as the realization of live edges forms a loop of size smaller than $2t+1$ in the vicinity of $i$, equation \eqref{eq:overcounting_sets} will be verified.

\begin{figure}[thb]
    \centering
    \includegraphics[width=0.37\columnwidth]{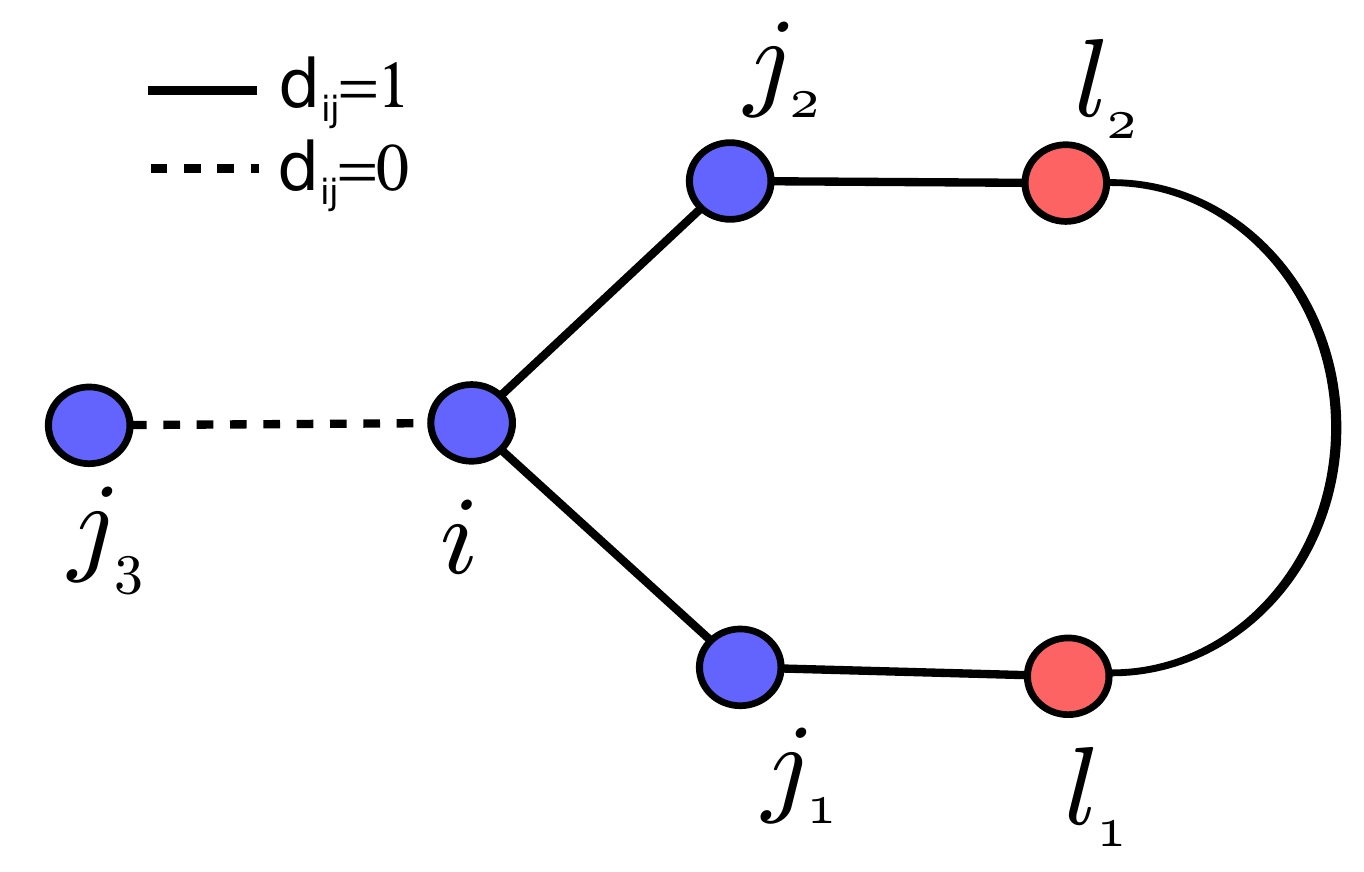}
    \caption{Example used for illustrating the statements of Theorem~\ref{th:bound}. Given the estimation for node $i$ at $t=3$ and for a given realization of live edges represented by $\mathbf{d}$, both red nodes will be counted twice in the right hand side of \eqref{eq:overcounting}, but only once in the left hand side of this expression. }
    \label{fig:reachability}
\end{figure}

As a consequence of \eqref{eq:overcounting_sets}, the following relation
\begin{equation}
    \left\langle \prod_{l \in \mathcal{N}^i_t[\mathbf{d}]} \left[1 - p_{l}(0) \right] \right\rangle \geq \left\langle \prod_{j \in \partial i} \prod_{l \in \mathcal{N}^{i \leftarrow j}_t[\mathbf{d}]} \left[1 - p_{l}(0) \right] \right\rangle
    \label{eq:overcounting}
\end{equation}
is naturally satisfied because for some $l$ (e.g. $l_1$ and $l_2$ in Figure~\ref{fig:reachability}) the terms $\left[1 - p_{l}(0) \right] \leq 1$ are counted several times.

The proof of the Theorem follows from the following chain of inequalities:
\begin{align}
    p_{i}(T) & \stackrel{(a)}= 1 - \left[ 1 - p_{i}(0) \right] \left\langle \prod_{l \in \mathcal{N}^i_T[\mathbf{d}]} \left[1 - p_{l}(0) \right] \right\rangle
    \label{eq:theorem_(a)}
    \\
    & \stackrel{(b)}\leq 1 - \left[ 1 - p_{i}(0) \right] \left\langle \prod_{j \in \partial i} \prod_{l \in \mathcal{N}^{i \leftarrow j}_T[\mathbf{d}]} \left[1 - p_{l}(0) \right] \right\rangle \\
    & \stackrel{(c)}\leq 1 - \left[ 1 - p_{i}(0) \right] \prod_{j \in \partial i} \left\langle \prod_{l \in \mathcal{N}^{i \leftarrow j}_T[\mathbf{d}]} \left[1 - p_{l}(0) \right] \right\rangle \\
    & \stackrel{(d)}= 1 - \left[ 1 - p_{i}(0) \right] \prod_{j \in \partial i} q_{j \rightarrow i}(T) \\
    & \stackrel{(e)}= 1 - \left[ 1 - p_{i}(0) \right] \prod_{j \in \partial i} \left[ 1 - b_{ji} p_{j \rightarrow i}(T-1) \right]
    \label{eq:theorem(e)}
    \\
    & \stackrel{(f)}\leq 1 - \left[ 1 - p_{i}(0) \right] \prod_{j \in \partial i} \left[ 1 - b_{ji} \widehat{p}_{j \rightarrow i}(T-1) \right] \\
    & \stackrel{(g)}= \widehat{p}_{i}(T),
    \label{eq:theorem_(g)}
\end{align}
where $(a)$ is simply the definition \eqref{eq:q_counting}; $(b)$ follows from \eqref{eq:overcounting}; $(c)$ follows from the application of Lemma~\ref{lemma:chebyshev}; $(d)$ is due to the definition of $q_{j \rightarrow i}(t)$ \eqref{eq:q_message_counting}; $(e)$ expresses the application of the relation \eqref{eq:q_through_p}; $(f)$ follows from Lemma~\ref{lemma:underestimation}; and $(g)$ is the definition of the marginal probability \eqref{eq:marginal_p}. The statement of the theorem immediately follows from the application of \eqref{eq:theorem_(a)}-\eqref{eq:theorem_(g)} in \eqref{eq:exact_influence}:
\begin{equation}
    \sigma_{T}(\mathbf{p_{0}}) = \sum_{i \in V} p_i(T) \leq \sum_{i \in V} \widehat{p}_i(T) = \widehat{\sigma}_{T}(\mathbf{p_{0}}).
\end{equation}
\end{proof}

\section{Proofs of technical Lemmas}
\label{app:proofs_lemmas}

\begin{proof}[Proof of Lemma \ref{lemma:chebyshev}]
    The proof for general $n$ and $m$ straightforwardly follows from a subsequent application of the well-known result for the expectation of the product of two comonotonic functions $f(x)$ and $g(x)$ \cite{Chebyshev}:
    \begin{equation}
    \left\langle f(x) g(x) \right\rangle \geq  \left\langle f(x) \right\rangle \left\langle g(x) \right\rangle
    \label{eq:chebyshev_two_functions}
    \end{equation}
    The simplest way to prove \eqref{eq:chebyshev_two_functions} consists in observing that due to the comonotonocity
    \begin{equation*}
        \left[ f(x) - f(y) \right] \left[ g(x) - g(y) \right] \geq 0
    \end{equation*}
    for any $x$ and $y$, and applying expectation to the last expression.
\end{proof}

\begin{proof}[Proof of Lemma \ref{lemma:underestimation}]

Starting from definition \eqref{eq:exact_p} for $p_{j \rightarrow i}(t)$, let us rewrite the probability $q^{(i)}_{j}(t)$ in the live-edge representation:
\begin{equation}
    q^{(i)}_{j}(t)  = \left\langle \prod_{l \in \mathcal{N}^{j,(i)}_t[\mathbf{d}]} \left[1 - p_{l}(0) \right] \right\rangle,
\end{equation}
where $\mathcal{N}^{j,(i)}_t[\mathbf{d}]$ is the generalization of the reachable set for $j$ in the cavity graph where node $i$ has been removed. The first steps of Lemma's proof closely follow \eqref{eq:theorem_(a)}-\eqref{eq:theorem(e)} while working with the reachable sets in the graph where $i$ is in cavity. Repeating the same arguments as in the in the proof of Theorem~\ref{th:bound}, we prove the following relation:
\begin{equation}
    p_{j \rightarrow i}(t) \leq 1 - \left[ 1 - p_{j}(0) \right] \prod_{l\in \partial j \backslash i}\left[ 1 - b_{lj} p_{l \rightarrow j}(t-1) \right].
    \label{eq:lemma_step}
\end{equation}
The rest of the proof follows by induction. Due to the initializations \eqref{eq:initialization}, we have $p_{j \rightarrow i}(t) = \widehat{p}_{j \rightarrow i}(t)$ for all $(i,j) \in E$. Let us assume that $p_{j \rightarrow i}(t-1) \leq \widehat{p}_{j \rightarrow i}(t-1)$. Then, substituting $\widehat{p}_{l \rightarrow j}(t-1)$ instead of $p_{l \rightarrow j}(t-1)$ in the right hand side of \eqref{eq:lemma_step}, we obtain
\begin{equation}
    p_{j \rightarrow i}(t) \leq \left[ 1 - p_{j}(0) \right] \prod_{l\in \partial j \backslash i}\left[ 1 - b_{lj} \widehat{p}_{l \rightarrow j}(t-1) \right].
    \label{eq:lemma_step1}
\end{equation}
But from the DMP update equation \eqref{eq:p}, the right hand side of \eqref{eq:lemma_step1} is equal to $\widehat{p}_{j \rightarrow i}(t)$, from which Lemma's statement follows: $p_{j \rightarrow i}(t) \leq \widehat{p}_{j \rightarrow i}(t)$.
\end{proof}

\end{document}